\documentclass[review,onefignum,onetabnum]{siamart171218}

\usepackage{geometry}
\usepackage{tcolorbox}
\usepackage{tikz} 
\usepackage{tikz-cd}
\usetikzlibrary{decorations.shapes}
\usetikzlibrary{shadows, positioning, calc}
\tikzset{decorate sep/.style 2 args={decorate,decoration={shape backgrounds,shape=circle,shape size=#1,shape sep=#2}}}

\usepackage{amsopn}
\usepackage{bbm}
\usepackage{amsfonts}
\usepackage{graphicx}
\usepackage{epstopdf}
\usepackage{algorithmic}
\usepackage{amsmath}
\usepackage{mathabx}
\usepackage{mathrsfs}
\usepackage{caption}
\usepackage{subcaption}
\usepackage{enumerate}
\usepackage{float}

\ifpdf
  \DeclareGraphicsExtensions{.eps,.pdf,.png,.jpg}
\else
  \DeclareGraphicsExtensions{.eps}
\fi

\newtheorem{prop}{Proposition}[section]
\newsiamremark{remark}{Remark}
\newsiamremark{hypothesis}{Hypothesis}
\crefname{hypothesis}{Hypothesis}{Hypotheses}
\newsiamthm{claim}{Claim}

\headers{Spread Option with Liquidity Adjustments}{Kevin S. Zhang, Traian A. Pirvu}

\title{Pricing Spread Option with Liquidity Adjustments\thanks{Preprint.
\funding{This work was funded by NSERC grant 5-36700.}}}

\author{
Kevin S. Zhang\thanks{Department of Mathematics and Statistics, McMaster University, Hamilton, ON, Canada.}
\and
Traian A. Pirvu\thanks{Department of Mathematics and Statistics, McMaster University, Hamilton, ON, Canada.}
}

\begin{document}

\maketitle

\begin{abstract}
We study the pricing and hedging of European spread options on
correlated assets when, in contrast to the standard framework and consistent with imperfect liquidity markets, the trading in the stock market has a direct impact on stocks prices. We consider a \textit{partial-impact} and a \textit{full-impact} model in which the price impact is caused by every trading strategy in the market. The generalized Black-Scholes pricing partial differential equations (PDEs) are obtained and analysed. We perform a numerical analysis to exhibit the illiquidity effect on the replication strategy of the European spread option. Compared to the Black-Scholes model or a partial impact model, the trader in the full impact model buys more stock to replicate the option, and this leads to a higher option price.
\end{abstract}



\section{Introduction}
\par
Spread option gives holders the right but not the obligation to purchase the difference between two assets, at a cost. $\big(S_1(T)-S_2(T)-k\big)^+$ (here $k$ is the strike of the option, A.K.A the cost). While it can be written on all varieties of asset such as equities, bonds, and currencies, it has a unique role in the commodity market. The commodity market consists of many sectors such as agriculture, energy, and petroleum. In the agricultural market, \textit{Crush Spread} Johnson et al. (1991) \cite{Johnson} allows for the exchanges of unrefined soybeans with a combination of soybean oil and soybean meal. In the energy market, \textit{Spark Spread} Girma and Paulson (1998) \cite{Girma} pays the spread between natural gas and electricity. In the petro-market, \textit{Crack Spread} provides utility for the differential between the price of crude oil and petroleum products. Information on various Crack Spreads can be found in the NYMEX Rulebook (2020) \cite{NYMEX}. 
\par
Valuation of spread option involves solving a two-dimensional Black–Scholes PDE of the form: 
\begin{align}
\label{SpreadPDE}
    \left\{
      \begin{aligned}
        rV&=V_t+\frac{1}{2}\sigma_1^2s_1^2V_{s_1s_1}+\rho\sigma_1\sigma_2s_1s_2V_{s_1s_2}+\frac{1}{2}\sigma_2^2s_2^2V_{s_2s_2}  \cr
        &+rs_1V_{s_1}+rs_2V_{s_2},   \cr
        V(T,s_1,s_2)&=(s_1-s_2-k)^+,\quad\text{with $0<s_1,s_2<\infty$, $0\leq{t}\leq{T}$}.
      \end{aligned}\right.
\end{align}
The terminal condition of the above equation reflects the maturity payoff of the Spread option. One should be aware a closed form solution of \eqref{SpreadPDE} does not exist. Rather, we rely on numerical approaches for spread option pricing. In fact, there have been various numerical methods developed on approximating the pricing for these options. For example, \cite{Hurd,Carmona,Dempster} are well established methods.
\par
Underpinning classical pricing theory is the assumption of perfect market liquidity. Namely, that trading activity has no effect on asset prices. The relaxation of this assumption will impact the asset prices, and subsequently the value of derivative contract on the assets. This difference in price against the classical model was named \textit{liquidity valuation adjustment} (LVA) Pirvu and Zhang (2020) \cite{zhangshuai}. Willmot (2000) \cite{Wilmott} studied various price impact models arising from different trading strategies. In particular, Pirvu et al. \cite{Pirvu}, \cite{Pirvu2} investigated spread option impact models under \textit{Delta Hedging} strategies. The study yielded both a \textit{partial impact} model and \textit{full impact} model. The key difference between the two aforementioned model is within the replication process. The partial impact model uses the Delta of the impactless model while the full impact model use the Delta of the model with price impact. As of a result, the full impact model to have a non-linear pricing PDE.
\par
In this paper, we further explore the LVA market model of Pirvu and Zhang (2020) \cite{zhangshuai}. The model consist of two risky assets whose prices are driven by a pair of stochastic differential equations (SDEs). The illiquid asset price is modified to include full or partial price impact from trading. This give rise to two distinctive pricing PDEs corresponding to the partial impact and full impact model. Existence and uniqueness of the asset price SDEs are established. The approach used to derive the option price PDE is the replicating portfolio methodology. The option price PDEs are investigated in both partial impact and full impact models. We present a novel technique to numerically solve these PDEs. Motivated by \cite{Sirignano,alaradi}, we apply the Deep Galerkin Method (DGM) for PDEs. It relies on approximating the solution of a PDE with a deep neural network. The network is trained to satisfy the PDE's differential operator, its initial condition, as well as the boundary conditions. Their numerical routine is not affected by the PDE dimension because it is mesh free. We performed numerical experiments and analyse the results. We learn from our numerical experiments that the option hedging strategy in the full impact model has a higher financing cost and this result in a higher option price when compared to a partial impact or no impact model.
\par
The remainder of this paper is organized as follows. Section \ref{sec:model} presents the financial market model. The partial impact and full impact are treated in Subsections \ref{sec:Partial} and 
\ref{sec:Full}. The DGM method is presented in Section \ref{sec:Num} and numerical experiments in Section 5. The paper ends with a conclusion and an appendix sections.
\section{Model Framework}
\label{sec:model}
In this study, we adopt the basic mode framework of \cite{zhangshuai} for exchange options, and apply modifications to suit spread options. Our model of a market is based on a filtered probability space $(\Omega,\lbrace\mathscr{F}_t\rbrace_{t \in [0,T]},\mathbbm{P} )$ that satisfies the usual conditions, and consists of two assets. The asset prices are assumed to follow a $\lbrace\mathcal{F}_t\rbrace_{t\geq 0}$-adaptive two-dimensional It\^{o}-process $\mathbf{S}(t)=(S_1(t),S_2(t)).$ The randomness in this market is driven by a two-dimensional Brownian motion $\mathbf{W}(t)=(W_1(t),W_2(t))$. The assets prices dynamics are given by the following stochastic differential equations
    \begin{align}
    \label{Eq:FLMMOG}
    \begin{aligned}
        dS_1(t)&=\mu_1(t) S_1(t) dt+\sigma_1S_1(t)dW_1(t)+\lambda(t,S_1,S_2)df(t,S_1,S_2),\\
        dS_2(t)&=\mu_2(t) S_2(t) dt+\sigma_2 S_2(t)dW_2(t),
    \end{aligned}
    \end{align}
where $\lambda (t,S_1,S_2)$ and $f(t,S_1,S_2)$ are a price impact and the corresponding trading strategy respectively. Dynamics of the price impact function is given by 
    \begin{align}
    \label{impact}
    \bar{\lambda}\big(t,s_1,s_2\big)=
    \begin{cases}
    \epsilon\big(1-e^{-\beta(T-t)^{\frac{3}{2}}}\big)\quad&\text{if}\quad\underline{S}<s_1,s_2<\overline{S},
    \\
    0\quad&\text{otherwise},
    \end{cases}
    \end{align}
where $\underline{S}$ and $\overline{S}$ represents a trading floor and cap of the assets respectively. This cause the trading price impact to be truncated within the floor and cap. As for the other parameters, $\epsilon$ is the price impact per share, and $\beta$ is a decaying constant. 
It is important to emphasize that $\bar{\lambda}(t,s_1,s_2)$ is used for numerical approximation. The theoretical $\lambda(t,s_1,s_2)$ is a function with bounded derivative, and can be obtained through standard mollifying of $\bar{\lambda}(t,s_1,s_2)$.
\par
The trading strategy $f(t,s_1,s_2)$ can be a large trader's
strategy (partial impact), or any trading strategy (full impact). As in \cite{Pirvu2}, the SDEs of \eqref{Eq:FLMMOG} can be rewritten as:
    \begin{align}
    \label{FLMMSDE}
    \begin{aligned}
        dS_1(t)&=\bar{\mu}_1\big(\mathbf{S}(t)\big)dt+\bar{\sigma}_{11}\big(\mathbf{S}(t)\big)dW_1(t)+\bar{\sigma}_{12}\big(\mathbf{S}(t)\big)dW_2(t), \\
        dS_2(t)&=\bar{\mu}_2(t)dt+\bar{\sigma}_{21}(t)dW_1(t)+\bar{\sigma}_{22}(t)dW_2(t),
    \end{aligned}
    \end{align}
    where the drift and diffusion functions are:
    \begin{align*}
    \label{FLMMfunction}
    \bar{\mu}_1(t,s_1,s_2)&=\frac{1}{1-\lambda{}f_{s_1}}\Big(\mu_1s_1+\lambda{}f_{t}+s_2\mu_2\lambda{}f_{s_2}+\frac{f_{s_1s_2}(\rho\sigma_1\sigma_2s_1s_2+\sigma_2^2s_2^2\lambda{}f_{s_2})}{1-\lambda{}f_{s_1}}
    \\
    &+\frac{f_{s_1s_1}(\sigma_1^2s_1^2+\sigma^2_2s_2^2\lambda^2f_{s_2}^2+2\rho\sigma_1\sigma_2s_1s_2\lambda{}f_{s_2})}{2\big(1-\lambda{}f_{s_1}\big)^2}+\frac{\sigma_2^2s_2^2f_{s_2s_2}}{2}\Big),
    \\
    \bar{\mu}_2(t)&=\mu_2 s_2,
    \\
    \bar{\sigma}_{11}(t,s_1,s_2)&=\frac{\sigma_1s_1}{1-\lambda{}f_{s_1}},\qquad\qquad\bar{\sigma}_{12}(t,s_1,s_2)=\frac{\sigma_2s_2\lambda{}f_{s_2}}{1-\lambda{}f_{s_1}},
    \\
    \bar{\sigma}_{21}(t)&=\sigma_2 s_2 \rho,\qquad\qquad\qquad\qquad\bar{\sigma}_{22}(t)=\sigma_2s_2\sqrt{1-\rho^2}.
    \end{align*}
In Section \ref{sec:Partial} and Section \ref{sec:Full} we specialize the above SDEs to reflect partial and full impact. For the interest of pricing and hedging, we look for a risk-neutral pricing probability measure. The existence of such a probability measure is shown in Theorem \ref{thm:RNM}.
\begin{theorem}[\textbf{Finite Liquidity Risk-Neutral Measure}]
\label{thm:RNM}
There exists a unique risk-neutral measure $\widetilde{\mathbbm{P}}$ for {the finite liquidity market model (FLMM),} given by
\begin{align*}
    \widetilde{\mathbbm{P}}(\mathcal{A})=\int_{\mathcal{A}}Z(\omega)d\mathbbm{P}(\omega)\text{ for all $\mathcal{A}\in\mathcal{F}_T$},
    \end{align*}
where
    \begin{align*}
    Z(t)=\exp\big(-\int_0^t\langle\,\mathbf{\Theta}(u),d\mathbf{W}(t)\rangle-\frac{1}{2}\int_0^t||\mathbf{\Theta}(u)||^2du\big),
    \end{align*}
with the vector-valued market price of risk generator process
    \begin{gather*}
    \mathbf{\Theta}\big(t,S_1(t),S_2(t)\big)=\frac{1}{\bar{\sigma}_{11}\bar{\sigma}_{22}-\bar{\sigma}_{12}\bar{\sigma}_{21}}
    \begin{bmatrix} 
    \bar{\sigma}_{22}&-\bar{\sigma}_{12}\\
    -\bar{\sigma}_{21}&\bar{\sigma}_{11}
    \end{bmatrix}
    \begin{bmatrix} 
    \bar{\mu}_1-r\\
    \bar{\mu}_2-r
    \end{bmatrix}.
    \end{gather*}
    \end{theorem}

    \begin{proof}
Please see \ref{FLRNM} in the Appendix Section.
\end{proof}
\par
The following theorem states the existence of a unique solution for the SDE system driving the asset prices, based on standard result.
\begin{prop}[\textbf{Finite Liquidity SDE Existence}]
\par
Suppose the diffusion functions $\bar{\sigma}_{11}(t,s_1,s_2)$ and $\bar{\sigma}_{12}(t,s_1,s_2)$ of \eqref{FLMMSDE} are uniformly Lipschitz  continuous in $s_1,s_2\in\big(\mathbbm{R^+}\big)^2$. Then, the SDE system
\eqref{FLMMSDE} under $\widetilde{\mathbbm{P}}$ becomes
    \begin{align}
    \label{thm:SolGirsanov}
    \begin{aligned}
    &dS_1(t)=rdt+\bar{\sigma}_{11}\big(\mathbf{S}(t)\big)d\widetilde{W}_1(t)+\bar{\sigma}_{12}\big(\mathbf{S}(t)\big)d\widetilde{W}_2(t),
    \\
    &dS_2(t)=rdt+\bar{\sigma}_{21}(t)d\widetilde{W}_1(t)+\bar{\sigma}_{22}(t)d\widetilde{W}_2(t).
    \end{aligned}
    \end{align}
 Furthermore, it has a unique strong solution.
\ref{thm:RNM}.
\end{prop}

\section{Market Specification}
The purpose of this section is to describe the full and partial impact model in a manner that is clear, concise while still maintaining some mathematical rigorousness. For each of the partial and full impact model, we not only proved the existence and uniqueness for the market SDEs, but also the PDEs for option pricing. We also explains the market conditions that admits full price impact.
\subsection{Partial Impact}
\label{sec:Partial}
\par
In the case of partial price impact, only the trading activities of big institutional traders will affect market prices. If we assume all these big players are delta hedgers with the delta from BS model, then the risk-neutral dynamics of FLMM SDE system becomes
    \begin{align}
    \label{FLMMSDERNPart}
    \begin{aligned}
    &dS_1(t)=rdt+\bar{\sigma}^*_{11}\big(\mathbf{S}(t)\big)d\widetilde{W}_1(t)+\bar{\sigma}^*_{12}\big(\mathbf{S}(t)\big)d\widetilde{W}_2(t),
    \\
    &dS_2(t)=rdt+\bar{\sigma}_{21}(t)d\widetilde{W}_1(t)+\bar{\sigma}_{22}(t)d\widetilde{W}_2(t),
    \end{aligned}
    \end{align}
and the diffusion function are:
    \begin{align*}
    \bar{\sigma}_{11}^*(t,s_1,s_2)=\frac{\sigma_1s_1}{1-\lambda V_{s_1s_1}^{(BS)}},\qquad\qquad\bar{\sigma}_{12}^*(t,s_1,s_2)=\frac{\sigma_2s_2\lambda V_{s_1s_2}^{(BS)}}{1-\lambda V_{s_1s_1}^{(BS)}},\nonumber
    \end{align*}
where $V_{s_1s_1}^{(BS)}$ and $V_{s_1s_2}^{(BS)}$ are second order Greeks of spread option, under the BS Model. 

As one can see the study of Greeks for spread option
becomes crucial in further analysing the model. There has been some extensive studies on this in Li and Deng (2008) \cite{Li}. We derive our Greeks analysis from a Fourier transform method developed by Hurd and Zhou (2010) \cite{Hurd}. More details on the Greeks are available in Appendix \ref{subsec:Greek}.
Based on this we establish the existence and uniqueness
of asset prices in the next Theorem.
\begin{theorem}[\textbf{Finite Liquidity Existence III}]
\label{theorem:FLMMexistIII}
\par
The SDE system \eqref{FLMMSDERNPart} of FLMM under partial impact has a unique solution.
\end{theorem}
\begin{proof}
Please refer to the Appendix Section \ref{subsec:FiniteExist3}.
\end{proof}
\par
The replicating portfolio argument is fundamental in the derivations of the BS type PDE characterizing the spread option price. In this scenario, the portfolio used for replication have two assets and one cash account. The full argument can be found in Pirvu and Yazdanian (2016) \cite{Pirvu}. The resulting PDE will be linear and parabolic.
\begin{align}
\label{FLMMPDEPart}
    \left\{
        \begin{aligned}
            rV&=V_t+\frac{V_{s_1s_1}}{2(1-\lambda{}V^{(BS)}_{s_1s_1})^2}\big(\sigma_1^2s_1^2+\sigma^2_2s_2^2\lambda^2(V^{(BS)}_{s_1s_2})^2  \cr
            &+2\rho\sigma_1\sigma_2s_1s_2\lambda{}V^{(BS)}_{s_1s_2}\big)+\frac{V_{s_1s_2}}{1-\lambda{}V^{(BS)}_{s_1s_1}}\big(\rho\sigma_1\sigma_2s_1s_2   \cr
            &+\sigma_2^2s_2^2\lambda V^{(BS)}_{s_1s_2}\big)+\frac{1}{2}V_{s_2s_2}\sigma_2^2s_2^2+rs_1V_{s_1}+rs_2V_{s_2} \cr
            V(T,s_1,s_2)&=h(s_1,s_2),\quad\text{with $0<s_1,s_2<\infty$, $0\leq{t}\leq{T}$}.
        \end{aligned}
    \right.
\end{align}
To show the above PDE emit a unique classical solution, one may refer to Chapter 4 of Friedman (1975) \cite{Friedman}. In fact, the PDE \eqref{FLMMPDEPart} yield a unique classical solution whenever $1-\lambda{}f_{s_1}$ satisfies condition $(3)$ of Theorem \ref{theorem:FLMMexistIII}.

\subsection{Full Impact}
\label{sec:Full}
\par
In the case of full price impact, the interaction of every market participant (big or small), has a direct impact on the price of market constituents. If we assume all these participants applies a delta hedge strategy, then option sensitivities will be directly incorporated into the asset prices. In our market model, the second order option sensitivities (Gamma) becomes apart of the diffusion function of the illiquid asset $S_1$. We showcase the risk-neutral dynamics of FLMM SDE system in \eqref{FLMMSDERNFULL}.
\begin{align}
    \label{FLMMSDERNFULL}
        \begin{aligned}
        &dS_1(t)=rdt+\bar{\sigma}^{**}_{11}\big(\mathbf{S}(t)\big)d\widetilde{W}_1(t)+\bar{\sigma}^{**}_{12}\big(\mathbf{S}(t)\big)d\widetilde{W}_2(t),
        \\
        &dS_2(t)=rdt+\bar{\sigma}_{21}(t)d\widetilde{W}_1(t)+\bar{\sigma}_{22}(t)d\widetilde{W}_2(t),
        \end{aligned}
    \end{align}
where the diffusion function are
    \begin{align*}
    \label{FLMMfunctionFull}
    \bar{\sigma}_{11}^{**}(t,s_1,s_2)=\frac{\sigma_1s_1}{1-\lambda V_{s_1s_1}},\qquad\qquad\bar{\sigma}_{12}^{**}(t,s_1,s_2)=\frac{\sigma_2s_2\lambda V_{s_1s_2}}{1-\lambda V_{s_1s_1}}.\nonumber
    \end{align*}

\par
We can apply the canonical two-asset replicating portfolio argument, the result is the following non-linear BS-like PDE:
\begin{align}
\label{FLMMPDEFull}
    \left\{
      \begin{aligned}
        rV&=V_t+\frac{V_{s_1s_1}}{2(1-\lambda{}V_{s_1s_1})^2}\big(\sigma_1^2s_1^2+\lambda^2V_{s_1s_2}^2\sigma^2_2s_2^2   \cr
        &+2\lambda{}V_{s_1s_2}\rho\sigma_1\sigma_2s_1s_2\big)+\frac{V_{s_1s_2}}{1-\lambda{}V_{s_1s_1}}\big(\rho\sigma_1\sigma_2s_1s_2+\lambda{}V_{s_1s_2}\sigma_2^2s_2^2\big) \cr
        &+\frac{1}{2}V_{s_2s_2}\sigma_2^2s_2^2+rs_1V_{s_1}+rs_2V_{s_2},   \cr
        V(T,s_1,s_2)&=h(s_1,s_2),\quad\text{with $0<s_1,s_2<\infty$, $0\leq{t}\leq{T}$}.
      \end{aligned}\right.
\end{align}
This non-linearity is a major contrast between the partial and full price impact model, it brings a challenge to the establishment of model existence and uniqueness.
\par
We manage to establish existence and uniqueness for the market model \eqref{FLMMSDERNFULL} by initially showing the PDE \eqref{FLMMPDEFull} has a certain class of smooth solutions. Then extending existence and uniqueness to the SDE system \eqref{FLMMSDERNFULL} in a similar manner as outlined in Theorem \ref{theorem:FLMMexistIII}. This procedure is 
\begin{theorem}[\textbf{Finite Liquidity Existence IV}]
\label{theorem:FLMMexistIV}
\par
The SDE system \eqref{FLMMSDERNFULL} of FLMM SDEs with full impact assumption has a strong solution under $\widetilde{\mathbbm{P}}$.
\end{theorem}
\begin{proof}
Please refer to the Appendix Section \ref{subsec:FiniteExist4}.
\end{proof}

\section{Deep Galerkin Method}
\label{sec:Num}
\par
The curse of dimensionality is a common issue when attempting to solve high dimensional PDEs (include some literature). In high dimensions, methods such as finite difference and element not only become costly, but are often unstable. Muti-asset option pricing PDEs are affected by the curse of dimensionality. The \textit{Deep Galerkin Method} (DGM), developed by Sirignano and Spiliopoulos (2018) \cite{Sirignano}, have the potential to address these issues. From a high level, DGM can be viewed as a deep learning approach to solve weak formulation problems on PDE, but without the need to construct a mesh.
\par
In this section, we describe the DGM approach to solving two-asset pricing PDEs. We also outlines the steps required to solve \eqref{SpreadPDE}, \eqref{FLMMPDEPart} and \eqref{FLMMPDEFull}. Then, these solutions are compared against various benchmarks and insights on price impacts are subsequently derived.
\subsection{Adaption of DGM in Option Pricing}
\par
Consider a two-asset option pricing PDE:
\begin{align}
\label{dgmpde}
    \left\{
        \begin{aligned}
            \mathcal{L}V(t,s_1,s_2)=0,\qquad&(t,s_1,s_2)\in[0,T]\times(\mathbbm{R}^+)^2,\quad(interior) \cr
            V(T,s_1,s_2)=h(s_1,s_2),\qquad&(s_1,s_2)\in(\mathbbm{R}^+)^2,\quad(terminal) \cr
            V(t,s_1,s_2)=g(t,s_1,s_2),\qquad&(t,s_1,s_2)\in[0,T]\times\mathcal{B}\quad(boundary). \cr
        \end{aligned}
    \right.    
\end{align}
For the Sobolev space $\mathcal{H}_0^1=\mathcal{H}_0^1\big([0,T]\times(\mathbbm{R}^+)^2\big)$, the equivalent weak formulation of \eqref{dgmpde} is:
\begin{align}
\label{weakdgm}
    \left\{
        \begin{aligned}
        \langle\mathcal{L}V,u\rangle = 0\qquad &\forall{u} \in \mathcal{H}_0^1, \cr
        \langle V-h,v\rangle = 0\qquad &\forall{v} \in \bar{\mathcal{H}}_0^1, \cr
        \langle V-g,w\rangle = 0 \qquad &\forall{w} \in \partial \mathcal{H}_0^1. \cr
        \end{aligned}
    \right. 
\end{align}
\par
The regular Galerkin Method would require the careful selection of a set of basis functions $(\phi_1,\phi_2,...,\phi_N)$, that characterizes a finite approximation space $E_N\subset\mathcal{H}_0^1$. A unique best approximation, $\hat{V}_N$ of $V$, can be determined by projecting the PDE onto $E_N$. By increasing the dimension of approximation space, projection theorem gives a unique best approximation $\hat{V}_i$ in each approximation space $E_i$. The result is a sequence of approximators $\big\{\hat{V}_i\big\}_{i=N,N+1,..}$, that converges to $V$ by the completeness of $\mathcal{H}_0^1$. The rigorous formulation of this method can be found in \cite{Hunter}.
\par
DGM deviates from the regular Galerkin Method by assuming a neural network $f(t,s_1,s_2;\mathbf{\theta}):\mathbbm{R}^3\to\mathbbm{R}$ has the potential to capture the behavior of $V$. The network is subsequently initialized and trained with information gathered from the domain of $V$. This requires the selection of a meaningful objective function. The suggested objective function in \cite{Sirignano} closely resembles the weak formulation of \eqref{weakdgm} with the $\mathcal{L}^2$ inner product. The choice of $\mathcal{L}^2$ norm can be supported by literature such as \cite{Ayati}. The construction of the objective function proceeds as follows. For unit vectors $u \in  \mathcal{H}_0^1, v \in \bar{\mathcal{H}}_0^1$, and $\partial \mathcal{H}_0^1$, we apply the Cauchy-Schwartz Inequality to the weak formulation equations in \eqref{weakdgm}. Next, we sum the resulting terms, thus producing the objective function defined by
\begin{align}
\label{eq:tloss}
    \mathcal{J} = \|\mathcal{L}V\|^2_{[0,T]\times(\mathbbm{R}^+)^2} + \| V-h\|^2_{(\mathbbm{R}^+)^2} + \| V-g \|^2_{[0,T]\times(\mathbbm{R}^+)^2}.
\end{align}
\par
During the implementation stage, distributions are selected to generate points in the domain. This gives rise to a $\mathcal{L}^2$ distributional norm $\|f(\mathbf{x})\|^2_{\mathcal{D},\phi}=\int_{\mathcal{D}}|f(\mathbf{x})|^2\phi(\mathbf{x})d\mathbf{x}$, where $\phi(\mathbf{x})$ is a probability density function on the domain. Thus, the choice of $\phi(\mathbf{x})$ will significantly impact the performance of this method. A suitable objective function can be reformulated for our option pricing problems as:
    \begin{align}
    \label{eq:loss}
        J(\mathbf{\theta})&=J_1(\mathbf{\theta})+J_2(\mathbf{\theta})+J_3(\mathbf{\theta}),
    \end{align}
where $J_i(\theta)$, for $i=1,2,3$, are defined as follows:
    \begin{gather}
    \label{eq:loss1}
        J_1(\mathbf{\theta})=\|\mathcal{L}f(t,s_1,s_2;\mathbf{\theta})\|^2_{[0,T]\times(\mathbbm{R}^+)^2,\phi_1}.
    \end{gather}
This objective function measures how well the network satisfies the pricing PDE's differential operator.
    \begin{gather}
    \label{eq:loss2}
        J_2(\mathbf{\theta})=\|f(T,s_1,s_2;\mathbf{\theta})-h(s_1,s_2)\|^2_{(\mathbbm{R}^+)^2,\phi_2}.
    \end{gather}
This objective function measures how closely the network resembles the payoff function at maturity.
    \begin{gather}
    \label{eq:loss3}
        J_3(\mathbf{\theta})=\|f(t,s_1,s_2;\mathbf{\theta})-g(t,s_1,s_2)\|^2_{[0,T]\times\mathcal{B},\phi_3}.
    \end{gather}
The last objective function characterizes boundary conditions, or artificially created boundaries from asymptotics. For European style option pricing PDEs, these boundaries exist when underlying prices reach 0. The asymptotic appears when a price cap is impose on the underlings.
\par
The training data are generated as a tuple $(x^{(i)},x^{(T)},x^{(b)})$, where $x^{(i)}=(t,s_1,s_2)\sim\phi_1$, $x^{(T)}=(T,s_1,s_2)\sim\phi_2$ and $x^{(b)}=(t,s_1^{(b)},s_2^{(b)})\sim\phi_3$. In particular, $x^{(i)}$, $x^{(T)}$ and $x^{(b)}$ are generated from the interior, terminal and boundary (or artificial boundary) of the PDE respectively. The generated data are used to compute the objective function \eqref{eq:loss}.
\par
In the next phase, we apply a gradient descent algorithm, in hope of eventually finding a set of parameters $\theta$ for $f(t,s_1,s_2;\mathbf{\theta})$ that will produce a minima for the objective function. In fact, Correia et al. (2019) \cite{alaradi} mentioned DGM is strictly an optimization problem. Validation is unnecessary because the objective function directly characterizes weak formulation of the PDE. This also means a network that produce zero-valued objective function is the analytical solution of the PDE.
\par
The network architecture adopted in \cite{Sirignano} contains 1 dense layer and 3 DGM layers, all embedded with $tanh$ activation function. We modify the structure and incorporate the $swish$ activation function \cite{ramach}. Detailed arguments on the effectiveness of using $swish$ may be found in \cite{chen}. A summary of different types of activation function used in our network architecture are included in Table \ref{tab:activation}.  

\begin{table}[H]
    \caption{Activation Functions}
    \centering
    \begin{tabular}{|c|c|} 
    \hline
    Sigmoid & $\sigma(x)=\frac{1}{1+e^{-x}}$ \\
    \hline
    Tanh & $\sigma(x)=\frac{e^{x}-e^{-x}}{e^{x}+e^{-x}}$ \\
    \hline
    Swish & $\sigma(x)=\frac{x}{1+e^{-x}}$ \\
    \hline
    \end{tabular}
    \label{tab:activation}
    \end{table}

Figure \ref{Graph:DGMNetwork} captures the DGM network structure.
    \begin{figure}[H]
    \caption{DGM Network Architecture}
    \label{Graph:DGMNetwork}
    \centering
    \begin{tikzpicture}
        [innode/.style={circle, draw=red, fill=red!20, very thick, minimum size=2mm},
        outnode/.style={circle, draw=green, fill=green!20, very thick, minimum size=2mm},
        textnode/.style={rectangle, draw=white, fill=white!20, very thick, minimum size=5mm}]
        \draw[-,very thick,line width=0.75mm] (-1.75, 5) -- (5.75, 5);
        \draw[->,very thick,line width=0.75mm] (-1.75, 5) -- (-1.75, 4);
        \draw[->,very thick,line width=0.75mm] (0.75, 5) -- (0.75, 4);
        \draw[->,very thick,line width=0.75mm] (3.25, 5) -- (3.25, 4);
        \draw[->,very thick,line width=0.75mm] (5.75, 5) -- (5.75, 4);
        \draw[->,very thick,line width=0.75mm] (-1, 2) -- (0, 2);
        \draw[->,very thick,line width=0.75mm] (1.5, 2) -- (2.5, 2);
        \draw[->,very thick,line width=0.75mm] (4, 2) -- (5, 2);
        \draw[->,very thick,line width=0.75mm] (6.5, 2) -- (7.5, 2);
        \node[innode] at (-1.75, 5) {$X$};
        \draw[yellow,rounded corners=10,fill=yellow!20] (-2.5,0) rectangle (-1,4) node[black, pos=.5,rotate=90] {Swish Dense Layer};
        \draw[blue,rounded corners=10,fill=blue!20] (0,0) rectangle (1.5,4) node[black, pos=.5,rotate=90] {DGM Layer};
        \draw[blue,rounded corners=10,fill=blue!20] (2.5,0) rectangle (4,4) node[black, pos=.5,rotate=90] {DGM Layer};
        \draw[blue,rounded corners=10,fill=blue!20] (5,0) rectangle (6.5,4) node[black, pos=.5,rotate=90] {DGM Layer};
        \node[outnode] at (7.9, 2) {$Y$};
        \node[textnode] at (7.9, 1.3) {Linear Output};
    \end{tikzpicture}
    \end{figure}
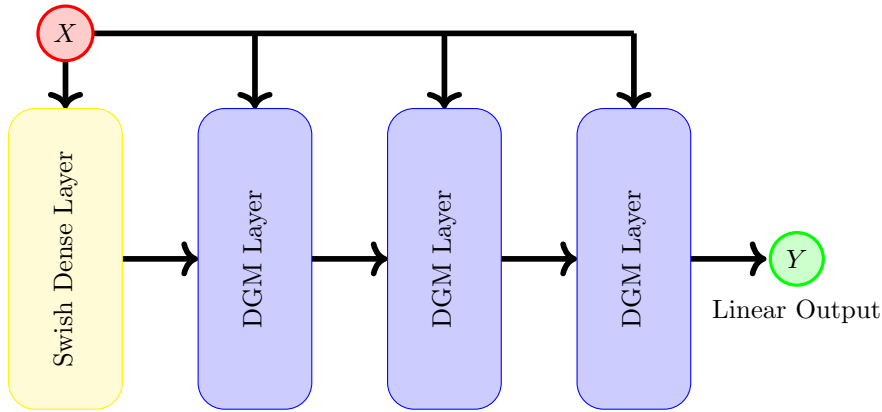
Our modified DGM layer is inspired by \textit{Gated Recurrent Unit} by Chung et al. \cite{Chung} (2014).  Figure \ref{Graph:MDGMLayer} captures the structure of each modified DGM layer.
    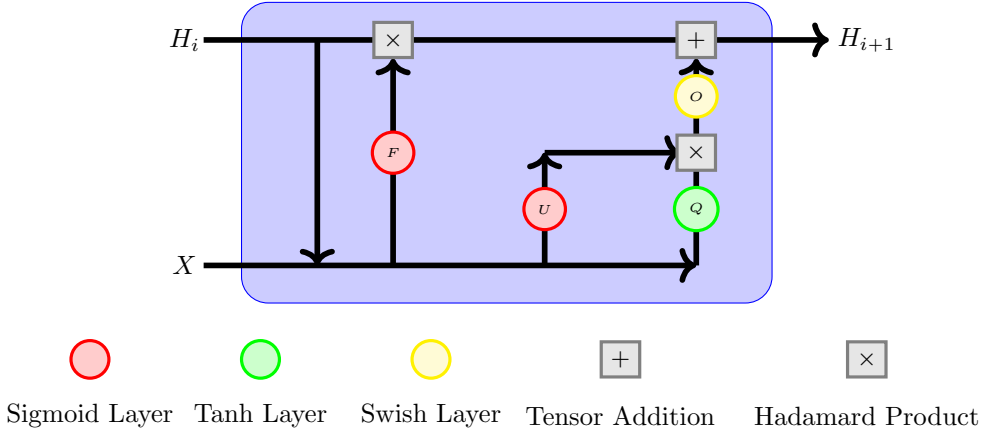
\begin{figure}[H]
    \caption{Modified DGM Layer}
    \label{Graph:MDGMLayer}
    \centering
    \begin{tikzpicture}
        [signode/.style={circle, draw=red, fill=red!20, very thick, minimum size=5mm},
        swishnode/.style={circle, draw=yellow, fill=yellow!20, very thick, minimum size=5mm},
        tannode/.style={circle, draw=green, fill=green!20, very thick, minimum size=5mm},
        opnode/.style={rectangle, draw=gray, fill=gray!20, very thick, minimum size=1mm},
        textnode/.style={rectangle, draw=white, fill=white!20, very thick, minimum size=5mm},]
        \node[textnode] at (-4.75, 0) {$H_i$};
        \node[textnode] at (4.25, 0) {$H_{i+1}$};
        \node[textnode] at (-4.75, -3) {$X$};
        \draw[blue,rounded corners=10,fill=blue!20] (-4,-3.5) rectangle (3,0.5) {};
        \draw[->,very thick,line width=0.75mm] (-4.5, 0) -- (3.75, 0);
        \draw[->,very thick,line width=0.75mm] (-3, 0) -- (-3, -3);
        \draw[->,very thick,line width=0.75mm] (-4.5, -3) -- (2, -3);
        \draw[->,very thick,line width=0.75mm] (-2, -3) -- (-2, -0.25);
        \draw[->,very thick,line width=0.75mm] (0, -3) -- (0, -1.5);
        \draw[->,very thick,line width=0.75mm] (0, -1.5) -- (1.8, -1.5);
        \draw[->,very thick,line width=0.75mm] (2, -3) -- (2, -0.25);
        \node[signode] at (-2, -1.5) {\tiny{$F$}};
        \node[signode] at (0, -2.25) {\tiny{$U$}};
        \node[tannode] at (2, -2.25) {\tiny{$Q$}};
        \node[swishnode] at (2, -0.75) {\tiny{$O$}};
        \node[opnode] at (2, 0) {$+$};
        \node[opnode] at (2, -1.5) {$\times$};
        \node[opnode] at (-2, 0) {$\times$};

        \node[signode] at (-6, -4.25) { };
        \node[textnode] at (-6, -5) {Sigmoid Layer};
        \node[tannode] at (-3.75, -4.25) {$ $};
        \node[textnode] at (-3.75, -5) {Tanh Layer};
        \node[swishnode] at (-1.5, -4.25) {$ $};
        \node[textnode] at (-1.5, -5) {Swish Layer};
        \node[opnode] at (1, -4.25) {$+$};
        \node[textnode] at (1, -5) {Tensor Addition};
        \node[opnode] at (4.25, -4.25) {$\times$};
        \node[textnode] at (4.25, -5) {Hadamard Product};
    \end{tikzpicture}
    \end{figure}
    
The mathematical operation behind our entire DGM network can be represented by the following set of equations:
\begin{align*}
    H_1 &= Swish\big(W_{0}X+b_0\big),
    \\
    F_l &= Sigmoid\big(W_{fx,l}X+W_{fh,l}H_l+b_{f,l}\big),\quad\text{for $l=1,2,3$},
    \\
    U_l &= Sigmoid\big(W_{ux,l}X+W_{uh,l}H_l+b_{u,l}\big),
    \\
    Q_l &= Tanh\big(W_{qx,l}X+W_{qh,l}H_l+b_{q,l}\big),
    \\
    O_l &= Swish\big(W_{ox,l}(U_l \circ O_l)+b_{o,l}\big).
    \\
    H_{l+1} &= F_l \circ H_l + O_l,
    \\
    Y &= H_{4} W_y + b_y,
\end{align*}
where $W$ are the weights, $b$ are the biases and $\circ$ is the Hadamard product.
\par
For iteration size $I$ and batch size $B$, we provide a general overview of the implementation of DGM in Algorithm \ref{algo:DGM}.
\begin{algorithm}[H]
\caption{Deep Galerkin Method for Option Pricing}
\label{algo:DGM}
\begin{algorithmic}
\STATE{Initialize learning rate $\alpha$ and network parameters $\mathbf{\theta}$}
\FOR{$i=1$ \textbf{to} $I$}
\STATE{Generate interior sample point $\mathbf{x}^{(i)}_{i}=[x^{(i)}_{i1},x^{(i)}_{i1},...,x^{(i)}_{iB}]$ from $\phi_1$}
\STATE{Generate terminal sample point $\mathbf{x}^{(T)}_{i}=[x^{(T)}_{i1},x^{(T)}_{i1},...,x^{(T)}_{iB}]$ from $\phi_2$}
\STATE{Generate boundary sample point $\mathbf{x}^{(b)}_{i}=[x^{(b)}_{i1},x^{(b)}_{i1},...,x^{(b)}_{iB}]$ from $\phi_3$}
\STATE{Compute the loss function:}
\STATE{\quad$J(\mathbf{\theta})=\|\mathcal{L}f(\mathbf{x}^{(i)}_{i};\mathbf{\theta})\|^2+\|f(\mathbf{x}^{(T)}_{i};\mathbf{\theta})-h(\mathbf{x}^{(T)}_{i})\|^2+\|f(\mathbf{x}^{(b)}_{i};\mathbf{\theta})-g(\mathbf{x}^{(b)}_{i})\|^2$}
\STATE{Take a descent step:}
\STATE{\quad$\mathbf{\theta}^{(new)}=\mathbf{\theta}^{(old)}-\alpha\frac{\partial J(\theta)}{\partial \mathbf{\theta}^{(old)}}$}
\STATE{Apply decay to the learning rate $\alpha$}
\ENDFOR
\end{algorithmic}
\end{algorithm}

\section{Experiments}
We run several experiments here to learn the partial \eqref{FLMMPDEPart} and full impact \eqref{FLMMPDEFull} option pricing PDEs. The adaptation of transfer learning is justified due to the similarity of these PDEs.

\subsection{Experiment Methodologies}
\label{subsect:ExpSetting}
\par
If an undergraduate student is given the task of learning graduate material. It is unlikely the student will perform very well. However, if that same student were to learn the prerequisites knowledge beforehand, and reattempt. That student certainly stand a better chance. In machine learning, this concept is often referred to as \textit{transfer learning}. It is the method of applying prior knowledge to related problems but often difficult to solve directly. Bengio (2012) \cite{Bengio} goes into extensive detail on transfer learning. Weiss et al. (2016) \cite{weiss} provides a formal definition for this method in terms of a domain $\mathcal{D}=\big\{\mathcal{X},\phi_{\mathcal{X}}\big\}$ and learning task $\mathcal{T}=\big\{\mathcal{Y},f(\cdot)\big\}$ ($\mathcal{X}$-feature space, $\phi_{\mathcal{X}}$-feature distribution, $\mathcal{Y}$-label space, $f(\cdot)$-predictive function).

\begin{definition}{(\textbf{Transfer Learning})}
\label{def:transfer}
\par
For a pair of domain and learning task $\mathcal{D}_s=\big\{\mathcal{X}_s,\phi_{\mathcal{X}_s}\big\}$, $\mathcal{T}_s=\big\{\mathcal{Y}_s,f_s(\cdot)\big\}$. Consider a target domain and learning task $\mathcal{D}_t=\big\{\mathcal{X}_t,\phi_{\mathcal{X}_t}\big\}$, $\mathcal{T}_t=\big\{\mathcal{Y}_t,f_t(\cdot)\big\}$. Transfer learning is the process of using relevant information of $f_s(\cdot)$ to improve the predictive capability of $f_t(\cdot)$.
\end{definition}
\par
By adopting transfer learning, we may train DGM nets to learn the FLMM pricing PDEs \eqref{FLMMPDEPart} and \eqref{FLMMSDERNFULL}. The aforementioned PDEs are special cases of the 2-dimensional BS PDE for Spread Option \eqref{SpreadPDE}. Therefore, we should train an initial DGM net to learn the relatively simpler BS PDE \eqref{SpreadPDE}. Subsequently, we may modify the objective function \eqref{eq:loss} in accordance to the more complicated PDE with price impacts, then further train the network to learn \eqref{FLMMPDEPart} and \eqref{FLMMSDERNFULL}.
\par
For some asset price cap $C$, we restrict the domain to the finite cube $[0,T]\times[0,C]^2$. This will allow us impose asymptotics as boundary conditions (see Section \ref{app:SpreadLoss} for more details), and in turn get a faster convergence. During implementation, we use mean squared error (MSE) as an estimator for the $\mathcal{L}^2$ norms in \eqref{eq:loss}. In calculation of MSE, $N$ is the mini-batch size for training, it should be large to ensure accuracy of the estimator.
\par
To implement DGM for the PDEs \eqref{SpreadPDE}, \eqref{FLMMPDEPart} and \eqref{FLMMPDEFull}, we can follow Algorithm \ref{algo:DGM} and define a distinct objective function for each of the PDEs. It is apparent the functions $J_2(\theta)$ and $J_3(\theta)$ are shared amongst these objective functions.
\begin{align*}
\label{SpreadLoss}
\hat{J}^{(b)}(\theta) &= \hat{J}^{(b)}_1(\theta) + \hat{J}_2(\theta) + \hat{J}_3(\theta),\qquad\text{(BS Model)}
\\
\hat{J}^{(p)}(\theta) &= \hat{J}^{(p)}_1(\theta) + \hat{J}_2(\theta) + \hat{J}_3(\theta),\qquad\text{(FLMM with partial impact)}
\\
\hat{J}^{(f)}(\theta) &= \hat{J}^{(f)}_1(\theta) + \hat{J}_2(\theta) + \hat{J}_3(\theta).\qquad\text{(FLMM with full impact)}
\end{align*}
Details on the objective functions $\hat{J}^{(b)}_1(\theta)$, $\hat{J}^{(p)}_1(\theta)$, $\hat{J}^{(f)}_1(\theta)$, $\hat{J}_2(\theta)$ and $\hat{J}_3(\theta)$ can be found in Section \ref{app:SpreadLoss}.
\par
The sampling method is completely problem depended, user should focus sampling from a sub-domain of highly probable option input parameters. In the case of spread option, we noticed the objective function convergence faster when we choose sampling distribution that produce more non-zero option value. We present details on our sampling distributions for each scenario in Table \ref{Tab:sample}.
    \begin{table}[H]
    \caption{Sampling Method}
    \centering
    \label{Tab:sample}
    \begin{tabular}{|c|c|c|} 
    \hline
     & \textbf{Support} & \textbf{Sampling Distribution} \\
    \hline
    $\hat{\phi}_{1}$ & $(t,s_1,s_2)\in[0,T]\times[0,C]^2$ & $t\sim\mathcal{U}(0,T)$, $s_1\sim C\beta(3,10)$, $s_2\sim C\beta(2,10)$  \\
    \hline
    $\hat{\phi}_{2}$ & $(s_1,s_2)\in[0,C]^2$ & $s_1\sim C\beta(3,10)$, $s_2\sim C\beta(2,10)$  \\
    \hline
    $\hat{\phi}_{31}$ & $(t,s_1)\in[0,T]\times[0,C]$ & $t\sim\mathcal{U}(0,T)$, $s_1\sim C\beta(3,10)$ \\
    $\hat{\phi}_{32}$ & $(t,s_2)\in[0,T]\times[0,C]$ & $t\sim\mathcal{U}(0,T)$, $s_2\sim C\beta(2,10)$ \\
    \hline
    \end{tabular}
    \label{tab:DGMSample}
    \end{table}

We present histograms for our sampling method in Figure \ref{Fig:Sample}. One may notice we sample $t$ uniformly, this because we desire to option prices evenly across a span of time to maturities. For the assets, we adopted two beta distribution, one slightly more centered than the other. The reason is because we desire the option value to be non-zero, then with high probability, the first asset should have a greater price than the second.

    \begin{figure}[H]
    \caption{Sampling Method}
    \label{Fig:Sample}
    \centering
    \begin{subfigure}[b]{0.475\textwidth}
    \includegraphics[height=4.1cm]{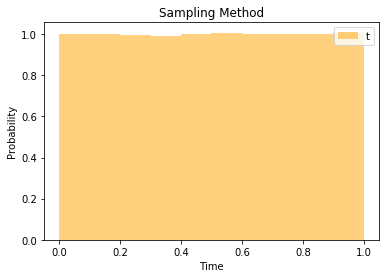}
    \end{subfigure}\quad
    \begin{subfigure}[b]{0.475\textwidth}
    \includegraphics[height=4.1cm]{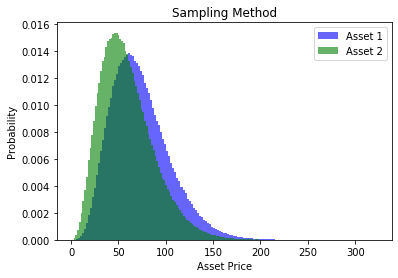}
    \end{subfigure}\quad
    \end{figure}

\subsection{Experiment Results}
\par
For all 3 PDEs (\eqref{SpreadPDE}, \eqref{FLMMPDEPart} and \eqref{FLMMPDEFull}), the shared option parameters we used are: $k=4$, $r=0.05$, $\rho=0.5$, $\sigma_1=0.4$, $\sigma_2=0.2$ and $T=1$. We are interested in the input region $(s_1,s_2)\in[0,100]^2$, the asset price cap is set to $C=600$. After training the first DGM network for \eqref{SpreadPDE}, we illustrate the results in Figure \ref{fig:DGMBSM}.

\begin{figure}[H]
\centering
\caption{Spread Option (BS Model)}
\label{fig:DGMBSM}
\begin{subfigure}[b]{0.475\textwidth}
\includegraphics[height=4.1cm]{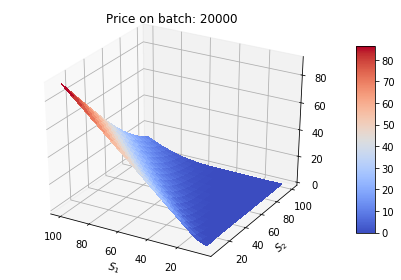}
\end{subfigure}\quad
\begin{subfigure}[b]{0.475\textwidth}
\includegraphics[height=4.1cm]{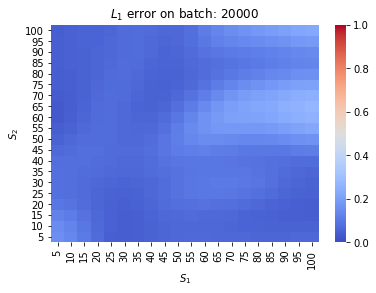}
\end{subfigure}\\
\begin{subfigure}[b]{0.475\textwidth}
\includegraphics[height=4.1cm]{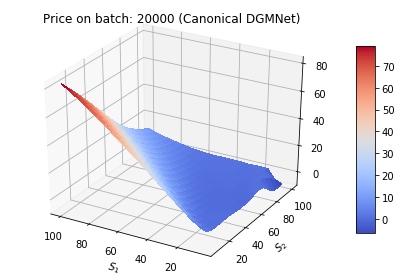}
\end{subfigure}\quad
\begin{subfigure}[b]{0.475\textwidth}
\includegraphics[height=4.1cm]{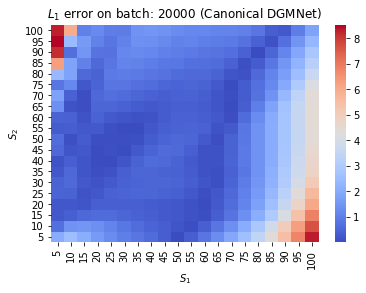}
\end{subfigure}
\footnotesize{*Benchmarked against FFT with grid size N=512.}
\end{figure}

The trained net matches extremely well with the FFT results. For this particular PDE, we our modified DGM net out performs the canonical DGM net. Although the trained net performs well on $[0,100]^2$, we should not expect the same level of performance will extend to $(\mathbbm{R}^+)^2$.
\par
Next, we take the previously trained model and apply transfer learning by switching the loss function to $\mathcal{L}^{(p)}(\theta)$ of \eqref{SpreadLoss}. The results are illustrated in Figure \ref{fig:Partial}.

\begin{figure}[H]
\centering
\caption{Spread Option (Partial Impact FLMM)}
\label{fig:Partial}
\begin{subfigure}[b]{0.475\textwidth}
\includegraphics[height=4.1cm]{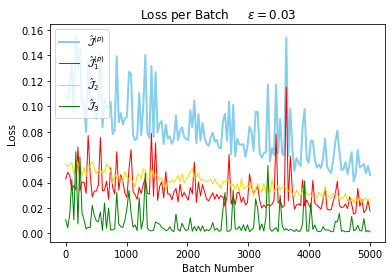}
\end{subfigure}\quad
\begin{subfigure}[b]{0.475\textwidth}
\includegraphics[height=4.1cm]{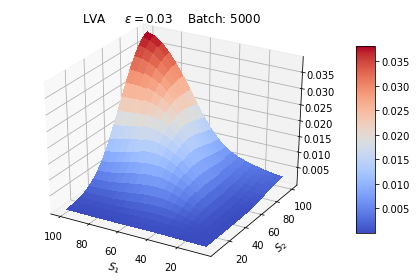}
\end{subfigure}\\
\begin{subfigure}[b]{0.475\textwidth}
\includegraphics[height=4.1cm]{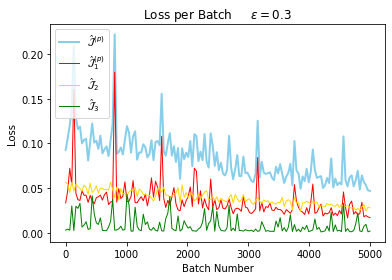}
\end{subfigure}\quad
\begin{subfigure}[b]{0.475\textwidth}
\includegraphics[height=4.1cm]{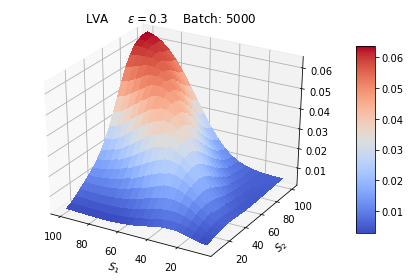}
\end{subfigure}
\footnotesize{*Benchmarked against FFT with grid size N=512.}
\end{figure}

From this experiment, we observe a price premium for the partial impact model. The premium is the result of illiquidity, it is the greatest for at-the-money options with high underlying asset prices. Furthermore, the liquidity premium increases as the cost-per-share parameter $\epsilon$ increases.
\par
For the full impact model, we fetch the pre-trained network for regular BS model and switch the loss function to $\mathcal{L}^{(f)}(\theta)$ of \eqref{SpreadLoss}. The results are illustrated in Figure \ref{fig:Full}.

\begin{figure}[H]
\centering
\caption{Spread Option (Full Impact FLMM)}
\label{fig:Full}
\begin{subfigure}[b]{0.475\textwidth}
\includegraphics[height=4.1cm]{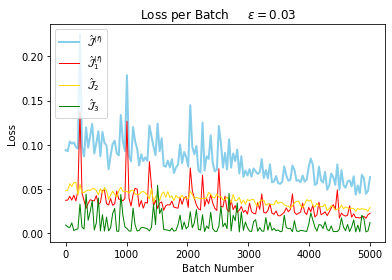}
\end{subfigure}\quad
\begin{subfigure}[b]{0.475\textwidth}
\includegraphics[height=4.1cm]{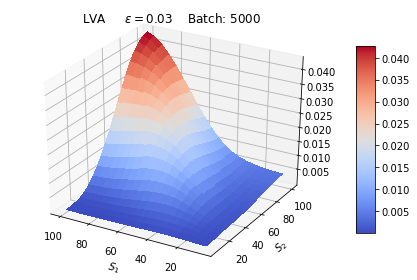}
\end{subfigure}\\
\begin{subfigure}[b]{0.475\textwidth}
\includegraphics[height=4.1cm]{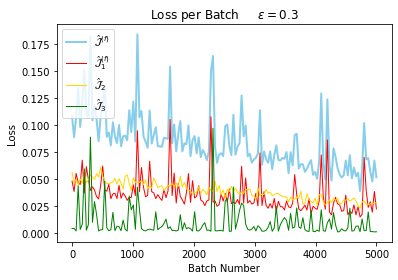}
\end{subfigure}\quad
\begin{subfigure}[b]{0.475\textwidth}
\includegraphics[height=4.1cm]{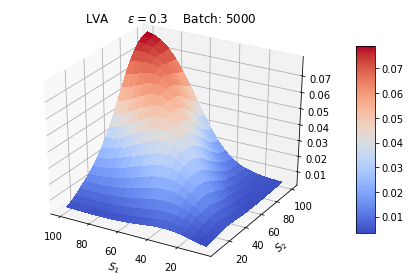}
\end{subfigure}
\footnotesize{*Benchmarked against FFT with grid size N=512.}
\end{figure}

From Figure \ref{fig:Full}, we observe a price premium for the full impact model. Once again, this price premium is the greatest for at-the-money options with high underlying asset prices. The liquidity premium also increases as the cost-per-share parameter $\epsilon$ increases.

\begin{figure}[H]
\centering
\caption{Spread Option (Partial vs Full)}
\label{fig:Diff}
\begin{subfigure}[b]{0.475\textwidth}
\includegraphics[height=4.1cm]{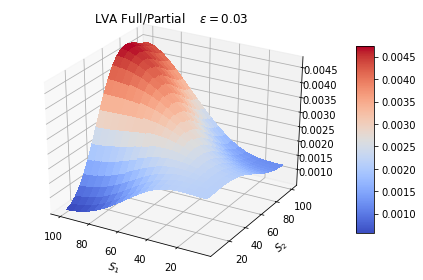}
\end{subfigure}\quad
\begin{subfigure}[b]{0.475\textwidth}
\includegraphics[height=4.1cm]{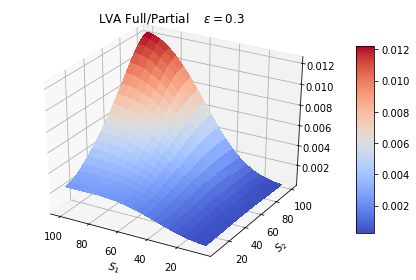}
\end{subfigure}
\end{figure}
\par
It's evident from Figure \ref{fig:Diff} the full impact model carries a greater liquidity premium than the partial impact mode. This is not a surprise because under full impact assumptions, all market trading activity has an impact on asset prices.

\section{Conclusions}
\label{sec:conclusions}
\par
We extended the FLMM of Pirvu and Zhang (2020) \cite{zhangshuai} to price spread options. We established the existence and uniqueness for the full and partial impact model SDEs driving the underlying asset prices and for the PDEs characterizing the spread option prices. We developed a variation of DGM method that essentially is a long short-term memory (LSTM) network with swish activation function, to numerically solve the option pricing PDEs. 
\par
Our DGM network has the ability to learn the PDEs \eqref{FLMMPDEFull} and \eqref{FLMMPDEPart}. The learning speed can be improved by learning the PDE of the impactless spread option \eqref{SpreadPDE} initially, then apply transfer learning. Our results indicate the full impact model requires greater liquidity value adjustment than the partial impact model. This finding is consistent because the full model takes includes all market trading activities. This paper may be useful for commodity traders who deal with illiquid underlying.
\section{Appendix}
\label{appendix}
\par
This section will include some of the formulas and proofs left out from the main body.

\subsection{Finite Liquidity Risk Neutral Measure}
\label{FLRNM}
\begin{proof}
\par Suppose there exists an equivalent measure $\widetilde{\mathbb{P}}$ generated by some process $\mathbf{\Theta}(t)$ such that under $\widetilde{\mathbb{P}}$, $dS_1(t)$ and $dS_2(t)$ has 
\begin{gather*}
Z(t)=\exp\big(-\int_0^t\langle\,\mathbf{\Theta}(u),d\mathbf{W}(t)\rangle-\frac{1}{2}\int_0^t||\mathbf{\Theta}(u)||^2du\big).
\end{gather*}
Then under $\widetilde{\mathbb{P}}$
\begin{align*}
d\widetilde{W}_1(t)&=dW_1(t)+\Theta_1(t)dt,\quad\text{and}
\\
d\widetilde{W}_2(t)&=dW_2(t)+\Theta_2(t)dt.
\end{align*}
Under $\widetilde{\mathbb{P}}$ we have the following dynamics:
\begin{align*}
dS_1(t)&=\Big(\bar{\mu}_1\big(\mathbf{S}(t)\big)-\bar{\sigma}_{11}\big(\mathbf{S}(t)\big)\Theta_1(t)-\bar{\sigma}_{12}\big(\mathbf{S}(t)\big)\Theta_2(t)\Big)dt\nonumber
+\bar{\sigma}_{11}\big(\mathbf{S}(t)\big)d\widetilde{W}_1(t)
\\
&+\bar{\sigma}_{12}\big(\mathbf{S}(t)\big)d\widetilde{W}_2(t),
\\
dS_2(t)&=\Big(\bar{\mu}_2(t)-\bar{\sigma}_{21}(t)\Theta_1(t)-\bar{\sigma}_{22}(t)\Theta_2(t)\Big)dt+\bar{\sigma}_{21}d\widetilde{W}_1(t)\nonumber
+\bar{\sigma}_{22}d\widetilde{W}_2(t).\nonumber
\end{align*}
Imposing the risk-less return rate under $\widetilde{\mathbb{P}}$ leads to the following linear system:
\begin{gather*}
\begin{bmatrix} 
\bar{\sigma}_{11}&\bar{\sigma}_{12}\\
\bar{\sigma}_{21}&\bar{\sigma}_{22}
\end{bmatrix}
\begin{bmatrix} 
\Theta_1(t)\\
\Theta_1(t)
\end{bmatrix}
=
\begin{bmatrix} 
\bar{\mu}_1-r\\
\bar{\mu}_2-r
\end{bmatrix}.
\end{gather*}
The system has a unique solution $\widetilde{\mathbb{P}}$ almost surely when the determinant is not zero, that is
$\bar{\sigma}_{11}(t)\bar{\sigma}_{22}(t)-\bar{\sigma}_{12}(t)\bar{\sigma}_{21}(t)\neq0,$ $\widetilde{\mathbb{P}}\otimes{d{t}}$ almost surely. Due to continuity of our processes, it will be sufficient to conclude for all $t$ we have $\bar{\sigma}_{11}(t)\bar{\sigma}_{22}(t)-\bar{\sigma}_{12}(t)\bar{\sigma}_{21}(t)\neq0,$ $\widetilde{\mathbb{P}}$ almost surely. Therefore, a necessary condition for the finite liquid market model to be complete is:
\begin{gather*}
\frac{S_1(t)}{S_2(t)}\neq\frac{\sigma_2\rho}{\sigma_1\sqrt{1-\rho^2}}\lambda(t, S_1(t),S_2(t))f_{s_2},\quad\widetilde{\mathbb{P}} \mbox{ almost surely.}
\end{gather*}
This condition is met in light of the continuous distribution of our processes.
\end{proof}

\subsection{Spread Option Greeks}
\label{subsec:Greek}
\par
Hurd and Zhou (2010) \cite{Hurd} pricing formula for Spread BS model. Let $\mathbf{x}=\big(\log(s_1),\log(s_2)\big)$ be the log initial asset prices, 
    \begin{align}
    \label{spreadprice}
    V^{(BS)}\big(t,s_1,s_2\big)&=\frac{ke^{-r(T-t)}}{(2\pi)^{2}}\int\int_{\mathbb{R}^2+i\epsilon}e^{i\mathbf{u}'\mathbf{x}}\hat{P}(\mathbf{u})\Phi_{\mathbf{x}}(\mathbf{u},\mathbf{\tau})d\mathbf{u},
    \end{align}
    where
    \begin{align*}
    \hat{P}(\mathbf{u})&=\frac{\Gamma\big(i(u_1+u_2)-1\big)\Gamma\big(-iu_2\big)}{\Gamma\big(iu_1+1\big)},\quad
    \Phi_{\mathbf{x}}(\mathbf{u},\tau)=\exp\big\{i\mathbf{u}'\big(r\mathbf{1}-\frac{tr(\Sigma)}{2}\big)\tau-\frac{\mathbf{u}'\Sigma\mathbf{u}\tau}{2}\big\}.
    \end{align*}
The $\epsilon=(\epsilon_1,\epsilon_2)$ term is a dampening factor with the restrictions $\epsilon_2>0$ and $\epsilon_1+\epsilon_2<-1$. To obtain a particular Greek, one can just differentiate equation \eqref{spreadprice} with respect to a desired parameter. We provide a summary of the Greeks in Table \ref{tab:spreadgreek}.

    \begin{table}[H]
    \caption{Spread Option Greeks}
    \centering
    \begin{tabular}{|c|} 
    \hline
    \textbf{First Order Greek} \\
    \hline
    $\begin{aligned}
        &\mathbf{\Delta}(t)=\begin{bmatrix} 
        \Delta_1\\
        \Delta_2
        \end{bmatrix}(t)=\frac{ke^{-r(T-t)}}{(2\pi)^{2}}\mathcal{T}^{(2)}\widebar{\mathbf{\Delta}}(t)=\frac{ke^{-r(T-t)}}{(2\pi)^{2}}
    \begin{bmatrix} 
    \frac{1}{s_1}\bar{\mathbf{\Delta}}_1\\
    \frac{1}{s_2}\bar{\mathbf{\Delta}}_2
    \end{bmatrix}(t) \\
    \end{aligned}$ \\
    \hline
    \textbf{Second Order Greek} \\
    \hline
    $\begin{aligned}
        \mathbf{\Gamma}(t)&=\begin{bmatrix} 
        \Gamma_{11}&\Gamma_{12}\\
        \Gamma_{21}&\Gamma_{22}
        \end{bmatrix}(t)=-\frac{ke^{-r(T-t)}}{(2\pi)^{2}}\Big(\mathcal{T}^{(3)}\widebar{\mathbf{\Delta}}(t)+\mathcal{T}^{(2)}\widebar{\mathbf{\Gamma}}(t)\mathcal{T}^{(2)}\Big)\\
        &=-\frac{ke^{-r(T-t)}}{(2\pi)^{2}}\begin{bmatrix} 
        \frac{1}{s_1^2}(\bar{\Delta}_{1}+\bar{\Gamma}_{11})&\frac{1}{s_1s_2}\bar{\Gamma}_{12}\\
        \frac{1}{s_1s_2}\bar{\Gamma}_{21}&\frac{1}{s_2^2}(\bar{\Delta}_{2}+\bar{\Gamma}_{22})
        \end{bmatrix}(t)  \\
    \end{aligned}$\\
    \hline
    \textbf{Third Order Greek} \\
    \hline
    $\begin{aligned}
        \mathbf{Spd}(t)&=\begin{bmatrix} 
        Spd_{111}&Spd_{112}\\
        Spd_{121}&Spd_{211}
        \end{bmatrix}
        \otimes\begin{bmatrix} 
        Spd_{122}&Spd_{221}\\
        Spd_{212}&Spd_{222}
        \end{bmatrix}(t)\\
        &=\frac{ke^{-r(T-t)}}{(2\pi)^2}\Big(2\mathcal{T}^{(4)}\widebar{\mathbf{\Delta}}(t)-\mathcal{T}^{(3)}\big(\mathcal{T}^{(2)}\widebar{\mathbf{\Gamma}}(t)\big)-\big(\mathcal{T}^{(2)}\big)^2\big(\widebar{\mathbf{Spd}}(t)\mathcal{T}^{(2)}\big)\Big)\\
        &=\frac{ke^{-r(T-t)}}{(2\pi)^{2}}\Big\{
        \begin{bmatrix} 
        \frac{1}{s_1^3}(2\bar{\Delta}_1+\bar{\Gamma}_{11}-\widebar{Spd_{111}})&-\frac{1}{s_1^2s_2}\widebar{Spd_{121}}\\
        \frac{1}{s_1^2s_2}(\bar{\Gamma}_{21}-\widebar{Spd_{211}})&-\frac{1}{s_1s_2^2}\widebar{Spd_{221}}
        \end{bmatrix},
        \\
        &
        \begin{bmatrix} 
        -\frac{1}{s_1^2s_2}\widebar{Spd_{112}}&\frac{1}{s_1^2s_2}(\bar{\Gamma}_{12}-\widebar{Spd}_{122})\\
        -\frac{1}{s_1s_2^2}\widebar{Spd_{212}}&\frac{1}{s_2^3}(2\bar{\Delta}_2+\bar{\Gamma}_{22}-\widebar{Spd}_{222})
        \end{bmatrix}
        \Big\}(t)
    \end{aligned}$\\
    \hline
    \end{tabular}
    \label{tab:spreadgreek}
    \end{table}

To derive Spread Option Delta, recall $\mathbf{x}=log{(\mathbf{s})}$ and
\begin{align*}
\mathbf{\Delta}(t)&=\frac{\partial{}V^{(BS)}(t,\mathbf{s})}{\partial{\mathbf{x}}}=\frac{\partial{}\mathbf{x}}{\partial{\mathbf{s}}}\frac{\partial{}V^{(BS)}(t,\mathbf{s})}{\partial{\mathbf{x}}}\nonumber,
\end{align*}
if we let
\begin{align*}
\mathcal{T}^{(2)}=\frac{\partial{}\mathbf{x}}{\partial{\mathbf{s}}}=
\begin{bmatrix} 
\frac{1}{s_1}&0\\
0&\frac{1}{s_2}
\end{bmatrix},
\end{align*}
then by taking the matrix derivative of \eqref{spreadprice}, we have
\begin{align*}
\mathbf{\Delta}(t)&=\frac{ke^{-r\tau}}{(2\pi)^{2}}\mathcal{T}^{(2)}\frac{\partial}{\partial{\mathbf{x}}}\Big(\int\int_{\mathbb{R}^2+i\epsilon}e^{i\mathbf{u'x}}\Phi(\mathbf{u},\tau)\hat{P}(\mathbf{u})d\textbf{u}\Big)\nonumber.
\end{align*}
By \textit{Dominated Convergence Theorem}, the order of differentiation and integration can be switched, and it follows
\begin{align}
\mathbf{\Delta}(t)&=\frac{ke^{-r\tau}}{(2\pi)^{2}}\mathcal{T}^{(2)}\int\int_{\mathbb{R}^2+i\epsilon}\frac{\partial}{\partial{\mathbf{x}}}\Big(e^{i\mathbf{u'X}(t)}\Big)\Phi(\textbf{u},\tau)\hat{P}(\mathbf{u})d\textbf{u}\nonumber
\\
&=\frac{ke^{-r\tau}}{(2\pi)^{2}}\mathcal{T}^{(2)}\int\int_{\mathbb{R}^2+i\epsilon}i\mathbf{u}e^{i\mathbf{u'x}}\Phi(\textbf{u},\tau)\hat{P}(\mathbf{u})d\textbf{u}\nonumber,
\end{align}
if we let
\begin{align}
\widebar{\mathbf{\Delta}}(t)&=\int\int_{\mathbb{R}^2+i\epsilon}i\mathbf{u}e^{i\mathbf{u'x}}\Phi(\textbf{u},\tau)\hat{P}(\mathbf{u})d\textbf{u},\nonumber
\end{align}
then we have
\begin{align*}
\mathbf{\Delta}(t)&=\frac{ke^{-r\tau}}{(2\pi)^{2}}\mathcal{T}^{(2)}\widebar{\mathbf{\Delta}}(t)=\frac{ke^{-r\tau}}{(2\pi)^{2}}
\begin{bmatrix} 
\frac{1}{s_1}\bar{\mathbf{\Delta}}_1\\
\frac{1}{s_2}\bar{\mathbf{\Delta}}_2
\end{bmatrix}(t)\nonumber.
\end{align*}

\par
The Spread Option Gamma can be defined as:
\begin{align*}
\mathbf{\Gamma}(t)&=\frac{\partial{}\mathbf{\Delta}(t)}{\partial{\mathbf{s}}}=\frac{ke^{-r(T-t)}}{(2\pi)^{2}}\Big(\frac{\partial{}\mathcal{T}^{(2)}}{\partial{\mathbf{s}}}\widebar{\mathbf{\Delta}}(t)+\mathcal{T}^{(2)}\frac{\partial\widebar{\mathbf{\Delta}}(t)}{\partial{\mathbf{x}}}\frac{\partial{}\mathbf{x}}{\partial{\mathbf{s}}}\Big),\nonumber
\end{align*}
let
\begin{align*}
\widebar{\mathbf{\Gamma}}(t)&=\int\int_{\mathbb{R}^2+i\epsilon}(\mathbf{u}\otimes\mathbf{u})e^{i\mathbf{u'x}}\Phi(\textbf{u},\tau)\hat{P}(\mathbf{u})d\textbf{u},\quad\text{and}
\\
\mathcal{T}^{(3)}&=\Big\{
\begin{bmatrix} 
\frac{1}{s_1^2}&0\\
0&0
\end{bmatrix}
,
\begin{bmatrix} 
0&0\\
0&\frac{1}{s_2^2}
\end{bmatrix}\Big\},
\end{align*}
then we have
\begin{align*}
\mathbf{\Gamma}(t)&=\frac{ke^{-r(T-t)}}{(2\pi)^{2}}\Big\{-\mathcal{T}^{(3)}\int\int_{\mathbb{R}^2+i\epsilon}i\mathbf{u}e^{i\mathbf{u'x}}\Phi(\textbf{u},\tau)\hat{P}(\mathbf{u})d\textbf{u}\nonumber
\\
&-\mathcal{T}^{(2)}\Big(\int\int_{\mathbb{R}^2+i\epsilon}(\mathbf{u}\otimes\mathbf{u})e^{i\mathbf{u'x}}\Phi(\textbf{u},\tau)\hat{P}(\mathbf{u})d\textbf{u}\Big)\mathcal{T}^{(2)}\Big\}\nonumber
\\
&=-\frac{ke^{-r(T-t)}}{(2\pi)^{2}}\Big(\mathcal{T}^{(3)}\widebar{\mathbf{\Delta}}(t)+\mathcal{T}^{(2)}\widebar{\mathbf{\Gamma}}(t)\mathcal{T}^{(2)}\Big)
\\
&=-\frac{ke^{-r(T-t)}}{(2\pi)^{2}}\begin{bmatrix} 
\frac{1}{s_1^2}(\bar{\Delta}_{1}+\bar{\Gamma}_{11})&\frac{1}{s_1s_2}\bar{\Gamma}_{12}\\
\frac{1}{s_1s_2}\bar{\Gamma}_{21}&\frac{1}{s_2^2}(\bar{\Delta}_{2}+\bar{\Gamma}_{22})
\end{bmatrix}(t)
\end{align*}

\par
The Spread Option Spd can be defined as:
\begin{align*}
\mathbf{Spd}(t)&=\frac{\partial{}\mathbf{\Gamma}(t)}{\partial{\mathbf{s}}}=-\frac{ke^{-r(T-t)}}{(2\pi)^{2}}\Big\{\frac{\partial{}\mathcal{T}^{(3)}}{\partial{\mathbf{s}}}\widebar{\mathbf{\Delta}}(t)+\mathcal{T}^{(3)}\frac{\partial\widebar{\mathbf{\Delta}}(t)}{\partial{\mathbf{x}}}\frac{\partial{}\mathbf{x}}{\partial{\mathbf{s}}}\nonumber
\\
&+\frac{\partial{}\mathcal{T}^{(2)}}{\partial{\mathbf{s}}}\widebar{\mathbf{\Gamma}}(t)\mathcal{T}^{(2)}+\mathcal{T}^{(2)}\Big(\frac{\partial\widebar{\mathbf{\Gamma}}(t)}{\partial{\mathbf{x}}}\frac{\partial{}\mathbf{x}}{\partial{\mathbf{s}}}\mathcal{T}^{(2)}+\widebar{\mathbf{\Gamma}}(t)\frac{\partial{}\mathcal{T}^{(2)}}{\partial{\mathbf{s}}}\Big)\Big\}\nonumber,
\end{align*}
let
\begin{align*}
\widebar{\mathbf{Spd}}(t)&=\int\int_{\mathbb{R}^2+i\epsilon}i(\mathbf{u}\otimes\mathbf{u}\otimes\mathbf{u})e^{i\mathbf{u'X}(t)}\Phi(\textbf{u},\tau)\hat{P}(\mathbf{u})d\textbf{u},\quad\text{and}
\\
\mathcal{T}^{(4)}&=\Big\{\Big(
\begin{bmatrix} 
\frac{1}{s_1^3}&0\\
0&0
\end{bmatrix}
,
\begin{bmatrix} 
0&0\\
0&0
\end{bmatrix}
\Big)
,
\Big(
\begin{bmatrix} 
0&0\\
0&0
\end{bmatrix}
,
\begin{bmatrix} 
0&0\\
0&\frac{1}{s_2^3}
\end{bmatrix}
\Big)\Big\},
\nonumber
\end{align*}
then we have,
\begin{align*}
\mathbf{Spd}(t)&=-\frac{ke^{-r\tau}}{(2\pi)^{2}}\Big\{-2\mathcal{T}^{(4)}\widebar{\mathbf{\Delta}}(t)+\mathcal{T}^{(3)}\widebar{\mathbf{\Gamma}}(t)\mathcal{T}^{(2)}-\mathcal{T}^{(3)}\widebar{\mathbf{\Gamma}}(t)\mathcal{T}^{(2)}\nonumber
\\
&+\mathcal{T}^{(2)}\Big(\big(\int\int_{\mathbb{R}^2+i\epsilon}i(\mathbf{u}\otimes\mathbf{u}\otimes\mathbf{u})e^{i\mathbf{u'X}(t)}\Phi(\textbf{u},\tau)\hat{P}(\mathbf{u})d\textbf{u}\big)\big(\mathcal{T}^{(2)}\big)^2-\widebar{\mathbf{\Gamma}}(t)\mathcal{T}^{(3)}\Big)\Big\}\nonumber
\\
&=\frac{ke^{-r\tau}}{(2\pi)^{2}}\Big(2\mathcal{T}^{(4)}\widebar{\mathbf{\Delta}}(t)+\mathcal{T}^{(2)}\widebar{\mathbf{\Gamma}}(t)\mathcal{T}^{(3)}-\mathcal{T}^{(2)}\widebar{\mathbf{Spd}}(t)\big(\mathcal{T}^{(2)}\big)^2\Big)
\\
&=\frac{ke^{-r\tau}}{(2\pi)^{2}}\Big\{
\begin{bmatrix} 
\frac{1}{s_1^3}(2\bar{\Delta}_1+\bar{\Gamma}_{11}-\widebar{Spd_{111}})&-\frac{1}{s_1^2s_2}\widebar{Spd_{121}}\\
\frac{1}{s_1^2s_2}(\bar{\Gamma}_{21}-\widebar{Spd_{211}})&-\frac{1}{s_1s_2^2}\widebar{Spd_{221}}
\end{bmatrix},
\\
&
\begin{bmatrix} 
-\frac{1}{s_1^2s_2}\widebar{Spd_{112}}&\frac{1}{s_1^2s_2}(\bar{\Gamma}_{12}-\widebar{Spd}_{122})\\
-\frac{1}{s_1s_2^2}\widebar{Spd_{212}}&\frac{1}{s_2^3}(2\bar{\Delta}_2+\bar{\Gamma}_{22}-\widebar{Spd}_{222})
\end{bmatrix}
\Big\}(t).
\end{align*}

\par
The method required to determine higher order Greeks becomes redundant. All the Greeks will be linear combinations of contour integral with the particular form:
\begin{align}
\label{contourgreek}
\widebar{\mathbf{Greek}}(t,s_1,s_2)=\int\int_{\mathbb{R}^2+i\epsilon}f_{\otimes}(\mathbf{u})e^{i\mathbf{u'x}}\Phi(\mathbf{u},\tau)\hat{P}(\mathbf{u})d\textbf{u},
\end{align}
where $f_{\otimes}(\mathbf{u})$ is some complex tensor polynomial function. For example, $\mathbf{\Gamma}(t)$ is a linear combination of the contour integrals $\bar{\mathbf{\Delta}}(t)$ and $\bar{\mathbf{\Gamma}}(t)$, with respective tensor polynomial functions $i\mathbf{u}$ and $\mathbf{u}\otimes\mathbf{u}$.

\subsection{Finite Liquidity Existence Theorem III}
\label{subsec:FiniteExist3}
\begin{proof}
\par
According to Proposition (\ref{thm:SolGirsanov}), the system of SDEs in \eqref{FLMMPDEPart} emit weak solutions when the diffusion functions 
\\
$\sigma_{11}^{*}(t,s_1,s_2)$ and $\sigma_{11}^{*}(t,s_1,s_2)$ are uniformly Lipshitz continuous. We can invoke Theorem 2.1 of Pirvu and Zhang (2020) \cite{zhangshuai}, and check whether the regularity conditions (1)-(3) are satisfied. We reinstate the conditions:
    \begin{align*}
    \begin{aligned}
    &(1)\qquad\|\lambda(s_1f_{s_1s_1}+s_1f_{s_1s_2}+f_{s_2}+s_2f_{s_2}+s_2f_{s_1s_1}+s_2f_{s_1s_2}+s_2f_{s_2s_2})\|<\infty, \cr
    &(2)\qquad\|\big(\lambda_{s_1}+\lambda_{s_2}\big)\big(s_1f_{s_1}+s_2f_{s_1}+s_2f_{s_2}\big)\|<\infty, \cr
    &(3)\qquad|||1-\lambda{}f_{s_1}|||>\delta_0, \text{ for some }\delta_0>0. \cr
    \end{aligned}
    \end{align*}
\par
To achieve this, first recall from \eqref{contourgreek} that all the Greeks are just linear combinations of the form:
\begin{align*}
\widebar{\mathbf{Greek}}(t,s_1,s_2)&=\int\int_{\mathbb{R}^2+i\epsilon}f_{\otimes}(\mathbf{u})e^{i\mathbf{u'x}}\Phi(\mathbf{u},\tau)\hat{P}(\mathbf{u})d\textbf{u}
\\
&=e^{-\epsilon'\mathbf{x}}\int\int_{\mathbb{R}^2}f_{\otimes}(\mathbf{u}+i\epsilon)e^{i\Re(\mathbf{u})'\mathbf{x}}\Phi(\mathbf{u}+i\epsilon,\tau)\hat{P}(\mathbf{u}+i\epsilon)d\textbf{u}
\\
&=\frac{1}{s_1^{\epsilon_1}s_2^{\epsilon_2}}\int\int_{\mathbb{R}^2}f_{\otimes}(\mathbf{u}+i\epsilon)e^{i\Re(\mathbf{u})'\mathbf{x}}\Phi(\mathbf{u}+i\epsilon,\tau)\hat{P}(\mathbf{u}+i\epsilon)d\textbf{u}
\\
&=\frac{1}{s_1^{\epsilon_1}s_2^{\epsilon_2}}\widebar{\mathbf{Greek}}^\Re(t,s_1,s_2).
\end{align*}
Here we use $\widebar{\mathbf{Greek}}^\Re(t,s_1,s_2)$ to distinguish between contour and real integrals forms. The term $e^{i\Re(\mathbf{u})'\mathbf{x}}$ lays on the complex unit circle, this results in $\|\widebar{\mathbf{Greek}}^\Re(t,s_1,s_2)\|<\infty$ for all Greeks. Then proving the regularity conditions only boils down to the terms $\frac{1}{s_1^{\epsilon_1}s_2^{\epsilon_2}}$.
\par
When we rewrite the counter integral as real integrals and substitute the BS Spread Greeks into Condition (1), we get:
    \begin{align*}
    &\lambda(t,s_1,s_2)\big(\frac{ke^{-r\tau}}{(2\pi)^{2}}\big)\big(\frac{s_1+s_2}{s_1^{3+\epsilon_1}s_2^{\epsilon_2}}(2\bar{\Delta}^\Re_1+\bar{\Gamma}^\Re_{11}-\widebar{Spd}^\Re_{111})-\frac{s_1+s_2}{s_1^{2+\epsilon_1}s_2^{1+\epsilon_2}}\widebar{Spd}^\Re_{112}
    \\
    &-\frac{1+s_2}{s_1^{1+\epsilon_1}s_2^{1+\epsilon_2}}\bar{\Gamma}^\Re_{12}-\frac{1+s_2}{s_1^{1+\epsilon_1}s_2^{1+\epsilon_2}}\widebar{Spd}^\Re_{122}\big).
    \end{align*}
By dropping the constants and bounded real integral terms, we have
    \begin{align}
    \label{order}
        \begin{aligned}
        &\lambda(t,s_1,s_2)\big(\frac{s_1+s_2}{s_1^{3+\epsilon_1}s_2^{\epsilon_2}}-\frac{s_1+s_2}{s_1^{2+\epsilon_1}s_2^{1+\epsilon_2}}-\frac{1+s_2}{s_1^{1+\epsilon_1}s_2^{1+\epsilon_2}}-\frac{1+s_2}{s_1^{1+\epsilon_1}s_2^{1+\epsilon_2}}\big)
        \\
        &=\lambda(t,s_1,s_2)\Big(\frac{s_2^2-3s_1^2-2s_1^2s_2}{s_1^{3+\epsilon_1}s_2^{1+\epsilon_2}}\Big).    
        \end{aligned}
    \end{align}
Since $\lambda(t,s_1,s_2)$ is only non-zero between $\underline{S}$ and $\overline{S}$, we conclude Expression \eqref{order} is bounded.
\par
Substitute the BS Spread Greeks into Condition (2) and adopting real integrals, we get:
    \begin{align*}
    &\big(\lambda_{s_1}+\lambda_{s_2}\big)\big(\frac{ke^{-r\tau}}{(2\pi)^{2}}\big)\big(\frac{s_1+s_2}{s_1^{2+\epsilon_1}s_2^{\epsilon_2}}(\bar{\Delta}^\Re_1+\bar{\Gamma}^\Re_{11})+\frac{s_2}{s_1^{1+\epsilon_1}s_2^{1+\epsilon_2}}\bar{\Gamma}^\Re_{12}\big).
    \end{align*}
By dropping the constants and bounded real integral terms, we have  
    \begin{align*}
    &\big(\lambda_{s_1}+\lambda_{s_2}\big)\big(\frac{2s_1s_2+s_2^2}{s_1^{2+\epsilon_1}s_2^{1+\epsilon_2}}\big).
    \end{align*}    
By the same argument of Condition (1), this term is also bounded.
\par
Condition (3) is bounded in $s_1$ and $s_2$ by the same logic as Condition (1) and (2). For the $t$ dimension of Condition $(3)$, $\Gamma_{11}(t,s_1,s_2)$ diverges as $t\to{T}$. By design $\lambda(t,s_1,s_2)$ has a higher order decaying that ensures $\lambda(t,s_1,s_2)\Gamma_{11}(t,s_1,s_2)$ to stay finite. Therefore we can always find a $\delta_0>0$ for which Condtion (3) holds.
\par
Since we have showed Condition $(1)$ to $(3)$ hold under $\widetilde{\mathbbm{P}}$ measure, we can conclude the system of SDE \eqref{FLMMPDEPart} emit strong solution under $\widetilde{\mathbbm{P}}$.
\end{proof}

\subsection{Finite Liquidity Existence and Uniqueness IV}
\label{subsec:FiniteExist4}
\par
\begin{proof}
By similar argument as \ref{theorem:FLMMexistIII},  the system of SDEs in \eqref{FLMMPDEFull} emits strong $\widetilde{\mathbbm{P}}$ solution when the diffusion functions $\sigma_{11}^{**}(t,s_1,s_2)$ and $\sigma_{11}^{**}(t,s_1,s_2)$ are uniformly Lipshitz continuous. Since these diffusion functions contain partial derivatives of $V(t,s_1,s_2)$ (the option with full impact), we need to establish existence of solution for the PDE \eqref{FLMMPDEFull}.
\par
Define $\Omega=\lbrace (t,x,y) \vert (t,x,y) \in [0, T ] \times (0,\infty) \times (0,\infty) \rbrace,$ and let $\mathcal{X}$ and $\mathcal{Y}$ be:
\begin{align*}
\mathcal{X}&=\big\{V \in C^{1,4,4}(\Omega) ~| ~s.t~\lambda V_{s_1s_1},~\lambda V_{s_1s_2},~\lambda V_{s_1s_1s_1}~and~\lambda V_{s_1s_1s_2}~are~bounded ~on~ \Omega,
\\&~and~conditions (1),(2),(3)~are~met\},
\\
\mathcal{Y}&=\Im \big(F(V(\varepsilon),\varepsilon)\big).
\end{align*}
 
Furthermore, take
    \begin{align*}
    F(V(\varepsilon),\varepsilon)&=V_t+\frac{V_{s_1s_1}}{2(1-\lambda{}V_{s_1})^2}\big(\sigma_1^2s_1^2+\lambda^2V_{s_2}^2\sigma^2_2s_2^2+2\lambda{}V_{s_2}\rho\sigma_1\sigma_2s_1s_2\big)
    \\
    &+\frac{V_{s_1s_2}}{1-\lambda{}V_{s_1}}\big(\rho\sigma_1\sigma_2s_1s_2+\lambda{}V_{s_2}\sigma_2^2s_2^2\big)+\frac{1}{2}V_{s_2s_2}\sigma_2^2s_2^2+rs_1V_{s_1}+rs_2V_{s_2}-rV.\nonumber
    \end{align*}
When we set $\varepsilon_0 =0$ and ${V_0} =V^{(BS)}$ (i.e solution of the BS PDE without price impact). According to \textit{Implicit Function Theorem}, given that the conditions
    \\
    \begin{enumerate}[\quad i.]
    \item $F({V_0}, \varepsilon_0) = 0$,
    \item the linear mapping $F_V({V_0}, \varepsilon_0) : \mathcal{X} \rightarrow \mathcal{Y},$ (the Gateaux derivative of $F$) is bijective,\\
    \end{enumerate}
are met, then there exists a neighborhood $V$ of ${V_0}$ and a neighborhood $\varepsilon$ of $\varepsilon_0$ such that for every $\varepsilon$ in that neighbourhood, there is a unique element $V(\varepsilon)$ such that
$F(V(\varepsilon),\varepsilon) = 0$. Moreover the mapping $\Lambda \ni \varepsilon \rightarrow V(\varepsilon)$ is of class $C^1$.
Condition (i) is trivially satisfied since $F({V_0},\varepsilon_0)=0$ (this is in fact the BS PDE without price impact). Next, we can show that $V_{0} \in X$ (the proof, based on standard arguments, is omitted). Now we are going to argue that the linear mapping $F_V({V_0},\varepsilon_0) : \mathcal{X} \rightarrow \mathcal{Y}$ is bijective. According to definition, the Gateaux derivative of $F$ at $V_0$ in the direction $V$ is
    \begin{align*}
    F_V({V_0},\varepsilon _0)V&=\lim_{\tau \rightarrow 0}\dfrac{F({V_0}+\tau V,\varepsilon_0)-F({V_0},\varepsilon_0)}{\tau}
    \\
    &=V_t+\frac{1}{2}\sigma_{1}^{2}s_{1}^{2}V_{s_{1}s_{1}}+\frac{1}{2}\sigma_{2}^{2}s_{2}^{2}V_{s_{2}s_{2}}+\sigma_{1}\sigma_{2}s_{1}s_{2}\rho V_{s_{1}s_{2}}
    \\
    &+rs_{1}V_{s_1}+rs_{2}V_{s_2}-rV.
    \end{align*}
Thus, the operator $\mathcal{L}=F_V({V_0},\varepsilon _0)$ is
\begin{align*}
\mathcal{L}=\frac{\partial }{\partial t}+\frac{1}{2}\sigma_{1}^{2}s_{1}^{2}\frac{\partial ^2}{\partial s_1^2}+ \frac{1}{2}\sigma_{2}^{2}s_{2}^{2}\frac{\partial ^2}{\partial s_2^2}+\sigma_{1}\sigma_{2}s_{1}s_{2}\rho \frac{\partial ^2}{\partial s_1\partial s_2}+rs_{1}\frac{\partial }{\partial s_{1}}+rs_{2}\frac{\partial }{\partial s_{2}}-r.
\end{align*}
Condition (ii) boils down to showing that the equation $\mathcal{L}V=g$ has a unique solution $V\in \mathcal{X}$ for every $g\in \mathcal{Y}.$ The proof of this, based on standard arguments, is omitted. The proof up to this point ensures the PDE \ref{FLMMPDEFull} has a solution in $C^{1,4,4}(\Omega)$. It remains to show the Lipshitz requirements of $\sigma_{11}^{**}(t,s_1,s_2)$ and $\sigma_{11}^{**}(t,s_1,s_2)$, which boils down to
\par
\begin{align*}
    (1)\qquad&\|\lambda(s_1V_{s_1s_1s_1}+s_1V_{s_1s_1s_2}+V_{s_1s_2}+s_2V_{s_1s_2}+s_2V_{s_1s_1s_1}+s_2V_{s_1s_1s_2}
    \\
    &+s_2V_{s_1s_2s_2})\|<\infty,
    \\
    (2)\qquad&\|\big(\lambda_{s_1}+\lambda_{s_2}\big)\big(s_1V_{s_1s_1}+s_2V_{s_1s_1}+s_2V_{s_1s_2}\big)\|<\infty,
    \\
    (3)\qquad&|||1-\lambda{}V_{s_1s_1}|||>\delta_0, \text{ for some }\delta_0>0.
    \end{align*}
Since $V(t,s_1,s_2)$ is $C^{1,4,4}(\Omega)$, conditions (1) and (2) are satisfied by the piece-wise property of $\lambda(t,s_1,s_2)$. Condition (3) is also satisfied because $\lambda$ has an order of $\mathcal{O}(\tau^{\frac{3}{2}})$, which has a decaying effect on  $|||1-\lambda{}V_{s_1s_1}|||$ as $t\to{T}$.
\end{proof}

\subsection{Spread Option DGM Loss Functions}
\label{app:SpreadLoss}
These functions are the MSE estimators for \eqref{eq:loss1}, \eqref{eq:loss2} and \eqref{eq:loss3}.
\begin{align*}
\hat{J}^{(b)}_1(\theta)&=\frac{1}{N}\sum^N_{(t,s_1,s_2)\sim\hat{\phi}_1}\big(\mathcal{L}^{(b)}f(t,s_1,s_2;\mathbf{\theta})\big)^2,
\\
\mathcal{L}^{(b)}&=\frac{\partial }{\partial t}+\frac{1}{2}\sigma_{1}^{2}s_{1}^{2}\frac{\partial ^2}{\partial s_1^2}+\sigma_{1}\sigma_{2}s_{1}s_{2}\rho \frac{\partial ^2}{\partial s_1\partial s_2}+\frac{1}{2}\sigma_{2}^{2}s_{2}^{2}\frac{\partial ^2}{\partial s_2^2}+rs_{1}\frac{\partial }{\partial s_{1}}
+rs_{2}\frac{\partial }{\partial s_{2}}-r,
\\
\hat{J}^{(p)}_1(\theta)&=\frac{1}{N}\sum^N_{(t,s_1,s_2)\sim\hat{\phi}_1}\big(\mathcal{L}^{(p)}f(t,s_1,s_2;\mathbf{\theta})\big)^2,
\\
\mathcal{L}^{(p)}&=\frac{\partial}{\partial t}+\frac{\sigma_1^2s_1^2+\sigma^2_2s_2^2\lambda^2(V^{(BS)}_{s_1s_2})^2+2\rho\sigma_1\sigma_2s_1s_2\lambda{}V^{(BS)}_{s_1s_2}}{2(1-\lambda{}V^{(BS)}_{s_1s_1})^2}\frac{\partial ^2}{\partial s_1^2}
\\
&+\frac{\rho\sigma_1\sigma_2s_1s_2+\sigma_2^2s_2^2\lambda V^{(BS)}_{s_1s_2}}{1-\lambda{}V^{(BS)}_{s_1s_1}}\frac{\partial ^2}{\partial s_1\partial s_2}+\frac{1}{2}\sigma_2^2s_2^2\frac{\partial ^2}{\partial s_2^2}+rs_1\frac{\partial}{\partial s_1}+rs_2\frac{\partial}{\partial s_2}-r,
\\
\hat{J}^{(f)}_1(\theta)&=\frac{1}{N}\sum^N_{(t,s_1,s_2)\sim\hat{\phi}_1}\big(\mathcal{L}^{(f)}f(t,s_1,s_2;\mathbf{\theta})\big)^2,
\\
\mathcal{L}^{(f)}&=\frac{\partial}{\partial t}+\frac{\big(\sigma_1^2s_1^2+\sigma^2_2s_2^2\lambda^2(\frac{\partial ^2}{\partial s_1\partial s_2})^2+2\rho\sigma_1\sigma_2s_1s_2\lambda{}\frac{\partial ^2}{\partial s_1\partial s_2}\big)}{2(1-\lambda{}\frac{\partial ^2}{\partial s_1\partial s_2})^2}\frac{\partial ^2}{\partial s_1^2}
\\
&+\frac{\rho\sigma_1\sigma_2s_1s_2+\sigma_2^2s_2^2\lambda \frac{\partial ^2}{\partial s_1\partial s_2}}{1-\lambda{}\frac{\partial ^2}{\partial s_1\partial s_2}}\frac{\partial ^2}{\partial s_1\partial s_2}+\frac{1}{2}\sigma_2^2s_2^2\frac{\partial ^2}{\partial s_2^2}+rs_1\frac{\partial}{\partial s_1}+rs_2\frac{\partial}{\partial s_2}-r,
\\
\hat{J}_2(\theta)&=\frac{1}{N}\sum_{(s_1,s_2)\sim\hat{\phi}_2}^N\big((s_1-s_2-k)^+-f(T,s_1,s_2;\mathbf{\theta})\big)^2,
\\
\hat{J}_3(\theta)&=\frac{1}{N}\sum_{(t,s_2)\sim\hat{\phi}_{32}}^N\big(C-s_2-ke^{-r(T-t)}-f(t,C,s_2;\mathbf{\theta})\big)^2
\\
&+\frac{1}{N}\sum_{(t,s_1)\sim\hat{\phi}_{31}}^N\big(s_1\mathcal{N}(d_+)-ke^{-r(T-t)}\mathcal{N}(d_-)-f(t,s_1,0;\mathbf{\theta})\big)^2
\\
&+\frac{1}{N}\sum_{(t,s_1)\sim\hat{\phi}_{31}}^Nf^2(t,s_1,C;\mathbf{\theta})+\frac{1}{N}\sum_{(t,s_2)\sim\hat{\phi}_{32}}^Nf^2(t,0,s_2;\mathbf{\theta}).
\end{align*}

\section*{Acknowledgments}
The authors are grateful to the anonymous referee for a careful checking of the details and for helpful comments that improved this paper.

\bibliographystyle{siamplain}
\bibliography{main}
\end{document}